\def\draft{1}  
\documentclass[11pt]{article}
\usepackage{amsmath}
\usepackage{amsfonts}
\usepackage{amssymb}
\usepackage{amsthm}
\usepackage{tikz}
\usepackage{tikz-cd}
\usepackage{mathtools}
\usepackage{graphicx}
\usepackage[letterpaper,hmargin=1in,vmargin=1in]{geometry}
\usepackage[usenames,dvipsnames]{pstricks}
\usepackage{epsfig}
\usepackage{pst-grad} 
\usepackage{pst-plot} 
\usepackage[space]{grffile} 
\usepackage{etoolbox} 
\makeatletter 
\patchcmd\Gread@eps{\@inputcheck#1 }{\@inputcheck"#1"\relax}{}{}
\makeatother 

\usepackage{ifthen}

\newcommand{\talkingPoint}[1]{\ifthenelse{\equal{\draft}{1}}{{\color{brown}{---#1---}}}{}}
\newcommand{\YSHI}[1]{\ifthenelse{\equal{\draft}{1}}{{\color{blue}{---YY:#1---}}}{#1}}
\newcommand{\commentout}[1]{}

\usepackage{hyperref}
\hypersetup{pdfpagemode=UseNone}

\usepackage{amsfonts}
\usepackage{amssymb}
\usepackage{amsmath}
\usepackage{latexsym}
\usepackage{amsthm}
\usepackage{color}
\usepackage{enumerate}
\usepackage{caption}

\newtheorem{Theorem}{Theorem} 

\newtheorem{Lemma}[Theorem]{Lemma}

\newtheorem{Proposition}[Theorem]{Proposition}

\newtheorem{Corollary}[Theorem]{Corollary}

\newtheorem{Definition}[Theorem]{Definition}

\DeclarePairedDelimiter\bra{\langle}{\rvert}
\DeclarePairedDelimiter\ket{\lvert}{\rangle}
\DeclareMathOperator{\Tr}{Tr}

\def\Tr{\textnormal{Tr}}

\def\<{\langle}
\def\>{\rangle}

\numberwithin{theorem}{section}
\numberwithin{equation}{section}

\title{Limitations on Transversal Computation through Quantum Homomorphic Encryption}

\author{%
   Michael Newman$^1$ and Yaoyun~Shi$^2$  \\
 \\
  $^1$Department of Mathematics\\
  \smallskip
  $^2$Department of Electrical Engineering and Computer Science\\
  University of Michigan, Ann Arbor, MI 48109, USA\\
  \texttt{{mgnewman@umich.edu, shiyy@umich.edu}}
}

\date{}

\begin{document}
\maketitle
\thispagestyle{empty}

\begin{abstract}
\noindent
Transversality is a simple and effective method for implementing quantum computation fault-tolerantly.
However, no quantum error-correcting code (QECC) can transversally implement a quantum universal gate set (Eastin and Knill, {\em Phys. Rev. Lett.}, 102, 110502). Since reversible classical computation is often a dominating part of useful quantum computation, whether or not it can be implemented transversally is an important open problem. We show that, other than a small set of non-additive codes that we cannot rule out, no binary QECC can transversally implement a classical reversible universal gate set.  In particular, no such QECC can implement the Toffoli gate transversally.  

We prove our result by constructing an information theoretically secure (but inefficient) quantum homomorphic encryption (ITS-QHE) scheme inspired by Ouyang {\em et al.} (arXiv:1508.00938).  Homomorphic encryption allows the implementation of certain functions directly on encrypted data, i.e. homomorphically.  Our scheme builds on almost any QECC, and implements that code's transversal gate set homomorphically.  We observe a restriction imposed by Nayak's bound  ({\em FOCS} 1999) on ITS-QHE, implying that any ITS quantum {\em fully} homomorphic scheme (ITS-QFHE) implementing the full set of classical reversible functions must be highly inefficient.  While our scheme incurs exponential overhead, any such QECC implementing Toffoli transversally would still violate this lower bound through our scheme.
\end{abstract}

\clearpage
\setcounter{page}{1}


\section{Introduction}

\subsection{Restrictions on transversal gates}

Transversal gates are surprisingly ubiquitous objects, finding applications in quantum cryptography \cite{Ouyang:2015}, \cite{Lai:2017}, quantum complexity theory \cite{Broadbent:2016}, and of course quantum fault-tolerance. The instability of quantum information is well-documented, and quantum error-correcting codes \cite{Knill:1996} allow the encoding of single qubits into multiple qubit systems so that errors on small subsets of physical qubits can be corrected \cite{Gottesman:1997}.  Performing computations on these codes carries the risk of propagating errors between different subsystems, unless the code can implement the computation in a way that preserves the subsystem structure.  Informally, these types of logical operators that decompose as a product across the subsystems are called \emph{transversal}, and the oft-cited Eastin-Knill theorem \cite{Eastin:2009}, \cite{Zeng:2007} limits the ability of quantum codes to prevent this error propagation.

\begin{Theorem}[Eastin-Knill]
No quantum error-correcting code can implement a quantum universal transversal gate set.
\end{Theorem}

These transversal gate sets are valuable as most models of fault-tolerant quantum computation implement associated transversal gate sets fault-tolerantly ``for free''.  Incurring comparatively significant overhead, often in the form of magic state distillation \cite{Fowler:2013}, \cite{Jones:2013}, gauge fixing \cite{Bombin:2013}, \cite{Poulsen:2016}, or more recently deconstructions of non-transversal gates into fault-tolerant pieces \cite{Yoder:2016}, one can fault-tolerantly implement some remaining gate set making the computation space universal.  Improving the efficiency of this overhead and designing new fault tolerant architectures to supplement transversal gates is central to quantum fault tolerance.   

Implementing fault-tolerant classical reversible computation efficiently would be extremely desirable as many quantum algorithms are primarily classical subroutines with a relatively small number of quantum gates, and there have been several proposals for doing so \cite{Paetznick:2013}, \cite{Cross:2008}, \cite{Jones:2012}.  For example, factoring a cryptographically large RSA key using Shor's algorithm requires around $3 \times 10^{11}$ Toffoli gates to perform modular exponentiation alone, and is the dominating portion of the circuit \cite{Fowler:2012}.  As Toffoli is universal for classical reversible computation, one might ask if there are any quantum error-correcting codes that can naturally implement Toffoli, and thus classical computations, transversally?  We give restrictions on the ability of QECCs to do this.

\begin{Theorem} [Informal]
Almost no quantum error-correcting code can implement a classical universal transversal gate set.  In particular, almost no quantum error-correcting code can implement the Toffoli gate transversally.
\end{Theorem}

The only exceptions to our theorem are non-additive distance $d$ codes that decompose as $d$-fold product states in their logical computational basis, where each ``subcode'' itself fails to be erasure-correcting.  Essentially, one can think of these as maximally redundant quantum codes: they are the concatenation of a repetition code with some distance $1$ inner code, similar to Shor's stabilizer code written as a $3$-fold product of $GHZ$ states.  We do not expect that any such code can implement Toffoli transversally, but it remains a case our proof technique cannot rule out.  In particular, our proof does apply to all binary additive codes.  The result is perhaps slightly surprising since there exist QECCs (e.g. triorthogonal codes) that can implement the $CCZ$ gate transversally \cite{Paetznick:2013}, and in fact transversal Toffoli gates can map between different quantum Reed-Solomon codes by increasing the degree of the underlying polynomial \cite{Cross:2008}.

\subsection{Quantum homomorphic encryption}

The main ingredient in our proof is an information-theoretically secure homomorphic encryption scheme.  Generally, homomorphic encryption \cite{Gentry:2009} is a means of delegating computation on sensitive data securely.  It allows for the encryption of data in such a way that another party can perform meaningful computation on the ciphertext {\it without} decoding, preserving the security of the underlying plaintext.  A scheme is termed {\em fully} homomorphic encryption (FHE) if it can implement a universal class of functions in some computation space.

Recently, extensions of homomorphic encryption to the quantum setting have been considered.  Instead of encrypting classical data and implementing addition and multiplication gates homomorphically, quantum homomorphic encryption aims to encrypt quantum data and implement unitary gates homomorphically.  Progress was made in \cite{Broadbent:2014}, and recently \cite{Dulek:2016} extended this work to a leveled scheme that could homomorphically implement all polynomial-sized quantum circuits.

The aforementioned schemes are only computationally secure, since they use classical FHE as a subroutine.  This is no great indictment: FHE is built on the difficulty of certain hard lattice problems that are leading candidates for quantum-secure encryption \cite{Cramer:2015}, \cite{Peikert:2015}.  However, quantum information often promises information-theoretic security (ITS) guarantees that are impossible classically.  Intermediate advances have also been made in this more restrictive setting.  One such scheme allows for the implementation of a large class of unitaries homomorphically, but with less stringent ITS guarantees \cite{Tan:2014}.  

More recently, \cite{Ouyang:2015} proposed a compact ITS-QHE scheme in which the size of the encoding scales polynomially with size of the input for the limited Clifford circuit class.  This scheme achieves the strongest notion of imperfect ITS, with the probability of distinguishing between any two ciphertexts exponentially suppressed in the size of the encoding.
 
The scheme in \cite{Ouyang:2015} is based on a ``noisy'' quantum encoding of the data.  They take an encoding circuit for a particular quantum code and replace the ancilla bits of the encoding with uniformly random noise.  Their encryption is then choosing a random embedding of this code into yet more uniformly random noise.  This scheme links ITS-QHE to transversal gates: the transversal gates for their code are exactly those gates that can be implemented homomorphically.

\subsection{Limitations on ITS-QHE}
There are fundamental limitations on what ITS homomorphic encryption can do.  It is known that for a purely classical scheme, efficient ITS-FHE is impossible, violating lower bounds in the setting of single server private information retrieval \cite{Fillinger:2012}.  It was further shown that in the best case scenario, when the mutual information between the plaintext and the ciphertext is precisely zero, efficient \emph{quantum} FHE is impossible \cite{Yu:2014}.  

This no-go result actually applies to the more restrictive setting of classical data being encrypted into quantum data, while allowing only classical reversible functions to be evaluated homomorphically.  Both \cite{Yu:2014} and \cite{Ouyang:2015} ask whether relaxing to imperfect ITS-security might allow for efficient ITS-QFHE.  Unfortunately, this is not the case. 

\begin{Proposition}\label{concurrent}[Informal]
Efficient ITS-QFHE is impossible.
\end{Proposition}

Concurrent to this work, this proposition was observed in \cite{Lai:2017}.  We provide a precise statement and proof of this restriction in Appendix \ref{AppendixITSQHE}.  This result can be seen by combining the proof technique in \cite{Fillinger:2012} with similar single server private information retrieval bounds in the quantum setting \cite{Baumeler:2013}. In essence, the inefficiency of ITS-QFHE follows from viewing ITS-QHE (on \emph{classical data} using a \emph{quantum encoding}) as a certain quantum random access encoding (QRAC) (see \cite{Ambainis:2008}) of the function class we wish to implement homomorphically.  Well-known bounds on QRACs \cite{Nayak:1999} place lower bounds on the encoding size of such a scheme, precluding efficiency.  Using a variant of the code-based ITS-QHE scheme proposed in \cite{Ouyang:2015}, we can then argue that (almost) any QECC implementing the Toffoli gate transversally would yield a scheme violating this lower bound. 

It is worth noting that similar tasks such as blind quantum computation \cite{Broadbent:2008} and computing on encrypted data \cite{Fisher:2013}, \cite{Broadbent:2015} allow ITS solutions, but they do so at the cost of interactivity between the client and server.  We do not allow this interactivity in our definition of homomorphic encryption.

\subsection{Comparison to related works}

The five works the most closely resemble our results are \cite{Eastin:2009},\cite{Zeng:2007} and \cite{Bravyi:2013}, which place restrictions on transversal gate sets for QECCs, and \cite{Ouyang:2015} and \cite{Lai:2017}, which use similar ITS-QHE constructions.  We very roughly summarize these results and compare them to our own.  

In \cite{Zeng:2007}, Zeng {\em et al.} were some of the first to place restrictions on quantum universal transversal gate sets for {\em additive} quantum codes by elucidating the stabilizer group structure.  Further work in \cite{Anderson:2014} classified the set of diagonal gates that can implement one and two qubit logical operations in stabilizer codes.  Shortly thereafter, \cite{Eastin:2009} showed that for {\em any} QECC, the transversal gate set must be finite, and so cannot approximate with arbitrary precision the full unitary group.  Intuitively, they make a Lie type argument by showing that infinitesimal transversal operations are themselves linear combinations of local error operators.  Since these unitaries must act identically on the codespace, it follows that the group of transversal operations must be finite.

More recently, \cite{Bravyi:2013} placed restrictions on the more general class of topologically protected logical gates in topological stabilizer codes, which include transversal gates as an optimal subset.  They showed that for a topological stabilizer code defined on a $d$-dimensional lattice, any such gate must lie in the $d$th level of the Clifford hierarchy.  These results were extended in \cite{Pastawski:2014} to more general stabilizer subsystem codes, and in Appendix \ref{Appendix3}, we detail how these arguments can be used to rule out classical reversible transversal computation for the subclass of stabilizer codes.

In the direction of ITS-QHE, \cite{Ouyang:2015} gave a compact and efficient ITS-QHE scheme for the restricted class of Clifford circuits.  Using magic state injection, they complete a universal gate set by adding the $T$-gate.  However, because the client and server cannot communicate during the protocol, they must limit themselves to circuits using a constant number of $T$-gates.  Again, their encryption is a ``noisy'' encoding of the data into some code followed by a secret embedding into random noise.  With this encryption, they are able to generate indistinguishable outputs using only polynomial overhead in the input size.  In the more recent work \cite{Lai:2017}, the authors independently observe Proposition~\ref{concurrent}.  They take the more positive approach of arguing what can be done in spite of this limitation, extending the ITS-QHE schematic in \cite{Ouyang:2015} to other {\em particular} error-correcting codes and using code concatenation to achieve security with only polynomial overhead.  This achieves ITS-QHE on the larger circuit class $IQP^+$, which is probably not classically simulable \cite{Bremner:2010}.

Because of the stringent lower bounds placed by Nayak, we actually forgo the noisy encoding circuit and embed QECCs directly into random noise after removing a correctable set of qubits.  This has the effect of increasing the overhead by an exponential factor in order to achieve security, but thanks to the roomy lower bound, this factor is still too small to allow an ITS-QFHE scheme.  

We can argue directly about the security of this scheme using the nonlocality of the quantum information being encoded in almost any QECC.  The idea is conceptually simple: in order to obtain encryptions of the data that are both secure and (sufficiently) short, we must inject randomness into the encodings themselves by withholding qubits from the code.  While ordinarily this would negatively affect the correctness of homomorphic evaluation, the error-correcting property allows us to inject this randomness while still maintaining perfect recoverability.  Then intuitively, spreading the information across the subsystems limits the complexity of the class of logical operators that don't couple the subsystems, i.e. the {\em transversal} operators.  This differs fundamentally from the approaches in $\cite{Zeng:2007}$ and $\cite{Eastin:2009}$ in that it is a quantitative information-type bound.  

It is not without its drawbacks however, as these maximally redundant codes fail to ``spread out'' the information sufficiently.  The prototypical example is Shor's code, which is the concatenation of a bit-flip and phase-flip code.  However, we can argue directly using the stabilizer group structure that no such {\em additive} code can implement Toffoli transversally.


\newcommand{\X}{\mathcal{X}}
\newcommand{\Y}{\mathcal{Y}}
\newcommand{\St}{\mathcal{S}}
\newcommand{\I}{\mathcal{I}}
\newcommand{\hro}{\hat{\rho}}
\newcommand{\Supp}{\mathrm{Supp}}
\newtheorem{conj}{Conjecture}

\section{Preliminaries}

\subsection{Quantum Information}

We quickly review some standard notation, followed by some less standard tools we will need from quantum information theory.  For a more complete view, see \cite{Nielsen:2011}.  

Throughout, we will be working with $2$-level qubit quantum systems. We denote by $|\mathcal{H}| = \log(\dim(\mathcal{H}))$, the number of qubits constituting state space $\mathcal{H}$.  We define a general quantum state to be a positive semi-definite operator $\rho \in L(\mathcal{H})$ of trace one.  We call such a state pure if rank$(\rho) = 1$, otherwise we call it mixed, and note that such an operator is mixed if and only if  $\Tr(\rho^2) < 1$.  For any operator $U \in \mathcal{H}_A$, we use the notation $U^A$ to indicate the operator $U_A \otimes I_B \in L(\mathcal{H}_A \otimes \mathcal{H_B})$.  When it is unclear which space a state lives in, we will denote its state space as a superscript (e.g. $\ket{\psi}^A$).  We also sometimes adopt the notation that for $\rho \in L(A \otimes B)$, $\rho^A = \Tr_B(\rho)$.  By slight abuse of notation, we also adopt the convention that for any permutation $\pi \in S_n$, $\pi$ can also indicate the unitary permutation operator corresponding to the physical permutation of qubits.  We also sometimes omit the dimension of an identity operator $I$, but usually the dimension is implicitly its trace normalization factor, e.g. $I/D$ acts on a space of dimension $D$.

The norm $\| \cdot \|_p$ refers the usual Schatten $p$-norm, so that for any $A \in L(\mathcal{H})$ with singular values $(a_1, \ldots, a_n)$, 

\[ \| A \|_p = \left(\sum\limits_{i=0}^n a_i^p\right)^{1/p} \text{ for $p > 1$, and } \hspace{.2 cm} \| A \|_1 = \sum\limits_{i=0}^n |a_i|. \] 

Further recall that

\[\frac{1}{2}\|\rho - \sigma\|_1 = \max\limits_{P\leq I} \Tr(P(\rho - \sigma)),\] 

and so we can think of the $1$-norm as a means of bounding the ability to distinguish two quantum states, where $\leq$ refers to the positive semidefinite partial ordering.  For a collection of quantum states $\{\rho_S\}$ indexed by $S \in \mathcal{S}$, we sometimes write $E_S[\rho_S]$ to denote the expectation over a uniformly random choice of $S$, $E_S[\rho_S] = \frac{1}{|\mathcal{S}|}\sum_{S \in \mathcal{S}} \rho_S$.  We will regularly be referring to several particular gates, and so list them here.

\[ X = \left( \begin{array}{cc}
0 & 1 \\
1 & 0
\end{array} \right) \hspace{2cm}
Y = \left( \begin{array}{cc}
0 & -i \\
i & 0
\end{array} \right) \hspace{2cm}
Z = \left( \begin{array}{cc}
1 & 0 \\
0 & -1
\end{array} \right)
\]

\[
CX = \left( \begin{array}{cccc}
1 & 0 & 0 & 0 \\
0 & 1 & 0 & 0 \\
0 & 0 & 0 & 1 \\
0 & 0 & 1 & 0
\end{array} \right) \hspace{2cm}
CZ = \left( \begin{array}{cccc}
1 & 0 & 0 & 0 \\
0 & 1 & 0 & 0 \\
0 & 0 & 1 & 0 \\
0 & 0 & 0 & -1
\end{array} \right)
\]

\[
\text{Toff} = \left( \begin{array}{cccccccc}
1 & 0 & 0 & 0 & 0 & 0 & 0 & 0 \\
0 & 1 & 0 & 0 & 0 & 0 & 0 & 0 \\
0 & 0 & 1 & 0 & 0 & 0 & 0 & 0 \\
0 & 0 & 0 & 1 & 0 & 0 & 0 & 0 \\
0 & 0 & 0 & 0 & 1 & 0 & 0 & 0 \\
0 & 0 & 0 & 0 & 0 & 1 & 0 & 0 \\
0 & 0 & 0 & 0 & 0 & 0 & 0 & 1 \\
0 & 0 & 0 & 0 & 0 & 0 & 1 & 0
\end{array} \right) \hspace{2cm}
CCZ = \left( \begin{array}{cccccccc}
1 & 0 & 0 & 0 & 0 & 0 & 0 & 0 \\
0 & 1 & 0 & 0 & 0 & 0 & 0 & 0 \\
0 & 0 & 1 & 0 & 0 & 0 & 0 & 0 \\
0 & 0 & 0 & 1 & 0 & 0 & 0 & 0 \\
0 & 0 & 0 & 0 & 1 & 0 & 0 & 0 \\
0 & 0 & 0 & 0 & 0 & 1 & 0 & 0 \\
0 & 0 & 0 & 0 & 0 & 0 & 1 & 0 \\
0 & 0 & 0 & 0 & 0 & 0 & 0 & -1
\end{array} \right)
\]

\begin{Definition}
\normalfont
An $(n,m,p)$-{\em quantum random access code} (QRAC) is a mapping from an $n$-bit string $x$ to an $m$-qubit quantum state $\rho_x$ along with a family of measurements $\{M_i^0,M_i^1\}_{i=1}^n$ satisfying, for all $x \in \{0,1\}^n$, $i \in [n]$,
\[\Tr(M_i^j\rho_x) \geq p \text{ if } x_i = j.\]
More generally, we can consider some protocol $M_i$ for retrieving $x_i$ and we call this protocol the $i$th query of the QRAC, satisfying $\Pr[M_i(\rho_x) = x_i] \geq p$.
\end{Definition}

\begin{Definition}\label{QECC}
\normalfont
An $n$-qubit {\em quantum code} is simply a subspace $C$ of an $n$-body Hilbert space $\mathcal{H}$ along with a fiducial orthonormal logical basis.  Let $P_C$ denote the projection onto $C$. The code is further called an $[[n,k,d]]$ {\em quantum error-correcting code} (QECC) if it is a subspace $C$ of dimension $2^k$ satisfying, for any $\lfloor\frac{d-1}{2}\rfloor$-local operators $E_a,E_b$, 
\[P_CE^\dag_aE_bP_C= \lambda_{ab}P_C.\]  
for some Hermitian $(\lambda_{ab})$.  We call the $E_a,E_b$ {\em correctable errors} and we say $d$ is the {\em distance} of the code.  

Operationally, this means that there exists a recovery channel $\mathcal{R}$ by which any $\lfloor\frac{d-1}{2}\rfloor$-local error (acting nontrivially on at most $\lfloor\frac{d-1}{2}\rfloor$ qubits) can be corrected.  Recall also that any $[[n,k,d]]$ QECC can correct up to $(d-1)$ errors in known locations.  We call these types of errors {\em erasure errors}, and so call codes satisfying $d\geq2$ {\em erasure-correcting codes}.
\end{Definition}

For simplicity, we will restrict our discussion to QECCs encoding a single qubit, i.e. $[[n,1,d]]$ QECCs.  A quick review of the proof shows that we can make this assumption without loss of generality by arguing against a classical universal transversal gate set on any single encoded logical qubit.  By similar reasoning, the proof applies to subsystem codes as well, where we require the logical operators of such a code to be independent of the state of the gauge qubits.

When we have a collection of $p$ logical qubits, {\em each} encoded into an $n$-qubit code, we can decompose the $np$ physical qubits into a fixed $n$-wise partition of $p$-qubits each so that every partitioning set contains exactly one qubit from each code block.  We refer to these partitioning sets as the subsystems of the collective code.  

\begin{Definition}
\normalfont
An $n$-qubit stabilizer group $\mathcal{S}$ is an abelian subgroup of the $n$-qubit Pauli group not containing $-I$.  An $n$-qubit {\em additive} (or stabilizer) {\em code} encoding $k$ logical qubits can be described as the simultaneous $(+1)$-eigenspace of the Pauli operators comprising an $n$-qubit stabilizer group $S$ with $n-k$ generators.  The logical Pauli operators of this code correspond to the normalizer cosets $\mathcal{N}(\mathcal{S})/\mathcal{S}$, and it follows that the distance of the code is the minimal weight operator in $\mathcal{N}(\mathcal{S})/\mathcal{S}$.
\end{Definition}

\begin{Definition}
\normalfont
For any quantum error correcting code $C$, we define its logical states $\ket{\psi}_L$ to be the physical encoding of $\ket{\psi}$ in $C$.  We define the logical gate $U_L$ to be a codespace preserving physical gate that satisfies, for all $\ket{\psi}_L \in C$, $U_L \ket{\psi}_L = (U\ket{\psi})_L$.  

The set of {\em transversal gates} $\mathcal{T}_C$ associated to $C$ are those logical gates that decompose as a product across the subsystems.  That is to say, $U_L \in \mathcal{T}_C$ if $U_L = U_1 \otimes \ldots \otimes U_n$, where $n$ is the length of the code, each $U_i$ acts on a single subsystem, and $U_L$ is a codespace preserving map on the code $C^{\otimes r}$ for $U$ an $r$-qubit gate.  We further define a logical gate to be {\em strongly transversal} if it decomposes as $U_L = U^{\otimes n}$.  Following the example of $\cite{Zeng:2007}$, we do not allow coordinate permutations in our definition of transversality.
\end{Definition}

\begin{Definition} \label{Definition8}
\normalfont
We say a quantum code $C =$ Span$_\mathbb{C}(\ket{\tilde{0}}, \ket{\tilde{1}})$ is an {\em $r$-fold code} if it can be written as 

\[\ket{i}_L = \bigotimes\limits_{j=1}^r\ket{\psi_{ij}}.\]

where each vector $\ket{\psi_{ij}}$ does not further decompose as a product state across any bipartition.  We additionally assume that $r \leq d$, that $ \ket{\psi_{0j}}$ and $\ket{\psi_{1j}}$ occupy the same subsystem, and that $\ket{\psi_{0j}} \perp \ket{\psi_{1j}}$.  It then makes sense to refer to Span$\{\ket{\psi_{0j}},\ket{\psi_{1j}}\}$ as the {\em $j$th subcode}. These assumptions are natural, and we justify them in our discussion.  

If the code is additionally an $[[n,1,d]]$ QECC with $r = d \geq 3$ and each subcode has distance $1$, we simply call the resulting code a {\em maximally redundant code}.  Note that any (pure state) code is at least a $1$-fold code.
\end{Definition}

The guiding example is Shor's code, which can be seen as the concatenation of a repetition outer code and a complementary $GHZ$ inner code, neither of which is quantum erasure correcting.  In the case that the subcodes are identical, any maximally redundant code is just the concatenation of a repetition code with some distance $1$ subcode.  Intuitively, these are codes for which you can't erase enough qubits to mix the state while still remaining perfectly correctable: while redundancy can be used in classical error-correcting codes to protect information, quantum error-correcting codes must ``spread out'' information to protect it.  In this sense, these codes are maximally redundant because they ``spread out'' the information the least.  We show that non-additive maximally redundant codes (i.e. maximally redundant codes for which the subcodes are comprised of non-stabilizer subspaces) are the only binary QECCs with the hope of implementing logical Toffoli transversally. 

\subsection{Homomorphic Encryption}

We define an ITS-QHE scheme as three algorithms performed between two parties which we will call Client and Server.  We restrict ourselves to the more limited setting of a quantum scheme implementing Boolean functions on classical data using quantum encodings.  Of course, any impossibility result then extends to the more difficult task of quantum computations on quantum inputs.  

The parameters of such a scheme are given by $(n,m,m',\epsilon,\epsilon')$ and some gate set $\mathcal{F}$.  Formally, we define the algorithms of an ITS-QHE scheme as acting on Client's private workspace $\mathcal{C}$, a message space $\mathcal{M}$ sent from Client to Server after encryption, and a message space $\mathcal{M}'$ sent from Server to Client after evaluation.

\begin{itemize}
\item[$(i)$] $QHE.Enc(x) = \rho^{\mathcal{CM}}$, in which the client chooses an $n$-bit input $x$ and encrypts with some private randomness to obtain $\rho$.  We assume that any quantum evaluation key is appended to the encryption. Client then sends the message portion of the encryption $\rho^{\mathcal{M}}$ to the Server.  We define $m$ to be the length of this message, the size of the encoding. \\

\item[$(ii)$] $QHE.Eval_f: L(\mathcal{M}) \longrightarrow L(\mathcal{M'})$, in which Server, with description of some circuit $f$, applies an evaluation map to an encrypted state, possibly consuming an evaluation key in the process.  Server then sends his portion of the state $\sigma^{\mathcal{CM}'}\coloneqq (I^\mathcal{C} \otimes QHE.Eval_f) (\rho^{\mathcal{CM}})$ back to Client, and we define the length of this message to be $m'$, the size of the evaluated encoding. \\

\item[$(iii)$] $QHE.Dec(\sigma^\mathcal{CM'}) = y$, in which Client decrypts the returned evaluated encoding using her side information $\mathcal{C}$ and recovers some associated plaintext $y$.  

\end{itemize}

\noindent As an encryption scheme, the above should certainly satisfy
\[QHE.Dec(QHE.Enc(x)) = x.\]
and as an ITS-QHE scheme, there are three additional properties the scheme should satisfy as well.

\begin{itemize}

\item[$(i)$] \emph{ $\epsilon$-Information-theoretic security}: for any inputs $x,y$, letting $\rho_x, \rho_y \in L(\mathcal{M})$ denote the outputs of $QHE.Enc$ on $\mathcal{M}$ (thinking of these states as mixtures of encryptions under uniformly random choices of secret key),
\[\|\rho_x - \rho_y\|_1 \leq \epsilon.\]

\item[$(ii)$] \emph{$\mathcal{F}$-homomorphic}: for any circuit $f \in \mathcal{F}$ and for any input $x$,
\[\Pr[QHE.Dec(QHE.Eval_f(QHE.Enc(x)))_1 \neq f(x)] \leq \epsilon' \]
where the probability is over the randomness of the protocol and the subscript $1$ denotes the first bit of the output.  This restriction to the first bit is just to argue directly about Boolean functions.  For ease of exposition, we also assume without loss of generality that our protocols are perfectly correct, and allow $\epsilon' = 0$.  We call the scheme fully homomorphic if it is homomorphic on the set of all classical Boolean circuits.\\

\item[$(iii)$] \emph{Compactness}: a priori, the server could do nothing except append a description of the circuit $f$ to be run by the decryption function after decrypting.  To avoid trivial solutions like this, we demand that the total time-complexity of Client's actions in the protocol do not scale with the complexity of the functions to be evaluated, but only with some fixed function on the size of the input.  Intuitively, this captures the motivation behind homomorphic encryption: limiting the computational cost to Client.  However, we note that the standard definition of compactness refers to the time-complexity of the decrypt function specifically.
\end{itemize}

We denote a scheme homomorphic for some class of functions $\mathcal{F}$ and satisfying all of these properties as an $\mathcal{F}$-ITS-QHE scheme.  If $\mathcal{F}$ is the set of all Boolean circuits, we denote such a scheme as an ITS-QFHE scheme.  We observe that such a scheme must be inefficient.  For a precise statement and proof, see Appendix \ref{AppendixITSQHE}.

\begin{Proposition}
The communication cost of ITS-QFHE must be exponential in the size of the input.
\end{Proposition}

\section{A coding based ITS-QHE scheme}

We now consider a strategy for implementing compact QHE using quantum codes.  This will be a simple ``block'' embedding encryption scheme homomorphically implementing \emph{quantum} circuits on \emph{classical} input, and is similar to the construction in \cite{Ouyang:2015}.  We will use the error-correcting property to withhold a correctable set of qubits from the encoding.

\begin{figure}
\center{\setlength{\fboxsep}{10pt}
\fbox{\parbox{5.2in}{
\textbf{Coding QHE Scheme:}
\vskip0.1in
\textit{Arguments:} 
\[
\begin{array}{rcl}
C & = & \textnormal{an $[[n+r,1,d]]$ $r$-fold QECC with $r < d$} \\
\vec{i} & \in & \{0,1\}^p \\
m & = & \textnormal{the size of each noise code block} \\
S & \in & [m]^n \textnormal{, the secret key}
\end{array}
\]

\begin{enumerate}
\item On input $\vec{i} \in \{0,1\}^p$, encode $\vec{i}$ as the pure state $\bigotimes\limits_{\ell = 1}^p \ket{i_\ell}_L$, for $\{\ket{0}_L, \ket{1}_L\}$ the logical computational basis defining $C$.
\vskip0.1in
\item Let $R$ be a collection of $r$ subsystems, each of $p$-qubits, comprised of one subsystem from each subcode.  Then form $\gamma^{\vec{i}} = \Tr_R(\bigotimes_\ell \ket{i_\ell}_L)$.  Essentially, $\gamma^{\vec{i}}$ is the state of the collection of codewords with each codeword missing one subsystem from each of its subcodes.
\vskip0.1in
\item Initialize $n$ ($p \times m$) arrays of maximally mixed qubits, and replace the $S_j$-th column of each array with the $j$-th subsystem of $\gamma^{\vec{i}}$.  This forms the encrypted state.
\vskip0.1in
\item Publish a constant number of labeled encryptions of $0$ and $1$, to be used as ancilla in homomorphic evaluation.
\end{enumerate}
}}
\caption{A description of the encryption procedure for the code based QHE scheme.}
\label{qhe}}
\end{figure}

  The scheme is detailed in Figures \ref{qhe} and \ref{pic}.  Using that notation to summarize, our encryption channel $\mathcal{E}$ is defined, for secret key $S$ and input string $\vec{i}$, as $\mathcal{E}(S, \vec{i}) = \gamma_{S}^{\vec{i}}$.  We sometimes use the notation $\gamma_S$ instead of $\gamma_S^{\vec{i}}$ or $\gamma$ instead of $\gamma^{\vec{i}}$, omitting $\vec{i}$ when we are unconcerned with the underlying plaintext.
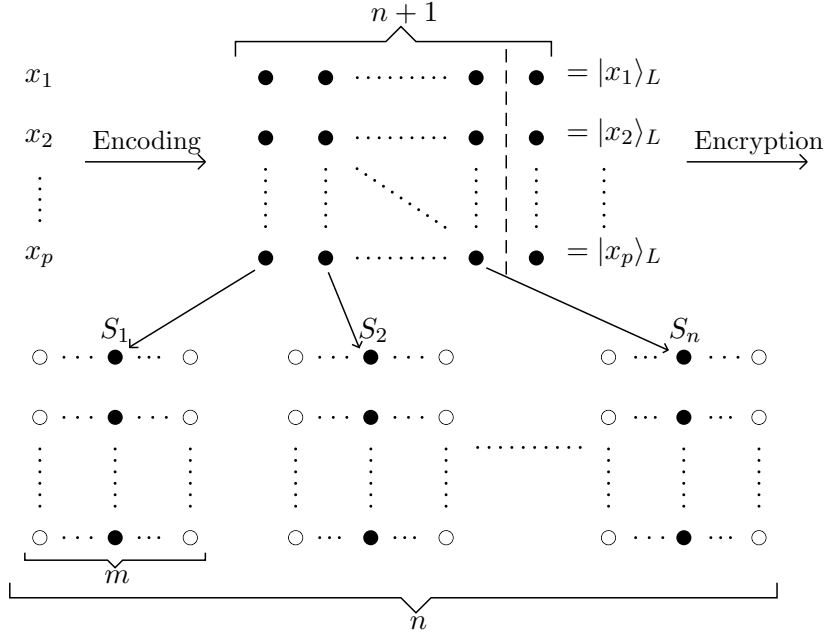
\begin{figure}
\center{\psscalebox{1.0 1.0} 
{
\begin{pspicture}(0,-4.155)(10.69,4.155)
\rput[bl](0.21,3.065){$x_1$}
\rput[bl](0.21,2.265){$x_2$}
\rput[bl](0.21,0.665){$x_p$}
\psdots[linecolor=black, dotsize=0.0](0.61,1.865)
\psdots[linecolor=black, dotsize=0.0](1.81,1.865)
\psdots[linecolor=black, dotsize=0.0](2.21,1.865)
\psdots[linecolor=black, dotsize=0.0](2.61,1.865)
\psline[linecolor=black, linewidth=0.02](1.01,2.065)(2.61,2.065)(2.61,2.065)
\psline[linecolor=black, linewidth=0.02](2.61,2.065)(2.51,2.165)
\psline[linecolor=black, linewidth=0.02](2.61,2.065)(2.51,1.965)
\psline[linecolor=black, linewidth=0.04, linestyle=dotted, dotsep=0.10583334cm](0.41,1.865)(0.41,1.265)
\psdots[linecolor=black, dotsize=0.2](3.41,3.185)
\psdots[linecolor=black, dotsize=0.2](4.21,3.185)
\psdots[linecolor=black, dotsize=0.2](6.21,3.185)
\psdots[linecolor=black, dotsize=0.2](3.41,2.385)
\psdots[linecolor=black, dotsize=0.2](4.21,2.385)
\psdots[linecolor=black, dotsize=0.2](6.21,2.385)
\psdots[linecolor=black, dotsize=0.2](3.41,0.785)
\psdots[linecolor=black, dotsize=0.2](4.21,0.785)
\psdots[linecolor=black, dotsize=0.2](6.21,0.785)
\psdots[linecolor=black, dotsize=0.2](7.01,3.185)
\psdots[linecolor=black, dotsize=0.2](7.01,2.385)
\psdots[linecolor=black, dotsize=0.2](7.01,0.785)
\psline[linecolor=black, linewidth=0.04, linestyle=dotted, dotsep=0.10583334cm](4.61,3.185)(5.81,3.185)
\psline[linecolor=black, linewidth=0.04, linestyle=dotted, dotsep=0.10583334cm](4.61,2.385)(5.81,2.385)
\psline[linecolor=black, linewidth=0.04, linestyle=dotted, dotsep=0.10583334cm](4.61,0.785)(5.81,0.785)
\psline[linecolor=black, linewidth=0.04, linestyle=dotted, dotsep=0.10583334cm](3.41,1.985)(3.41,1.185)(3.41,1.185)
\psline[linecolor=black, linewidth=0.04, linestyle=dotted, dotsep=0.10583334cm](4.21,1.985)(4.21,1.185)
\psline[linecolor=black, linewidth=0.04, linestyle=dotted, dotsep=0.10583334cm](6.21,1.985)(6.21,1.185)
\psline[linecolor=black, linewidth=0.04, linestyle=dotted, dotsep=0.10583334cm](7.01,1.985)(7.01,1.185)
\psline[linecolor=black, linewidth=0.04, linestyle=dotted, dotsep=0.10583334cm](4.61,1.985)(5.81,1.185)
\rput[bl](7.41,3.065){$=|x_1\rangle_L$}
\rput[bl](7.41,2.265){$=|x_2\rangle_L$}
\rput[bl](7.41,0.665){$=|x_p\rangle_L$}
\psline[linecolor=black, linewidth=0.04, linestyle=dotted, dotsep=0.10583334cm](7.91,1.985)(7.91,1.185)
\psline[linecolor=black, linewidth=0.02](3.01,3.465)(3.01,3.665)(5.01,3.665)(5.21,3.865)(5.41,3.665)(7.21,3.665)(7.21,3.465)
\rput[bl](4.81,3.905){$n+1$}
\psline[linecolor=black, linewidth=0.02](9.01,2.065)(10.61,2.065)(10.61,2.065)
\psline[linecolor=black, linewidth=0.02, linestyle=dashed, dash=0.17638889cm 0.10583334cm](10.61,2.065)(10.51,2.165)
\psline[linecolor=black, linewidth=0.02, linestyle=dashed, dash=0.17638889cm 0.10583334cm](10.61,2.065)(10.51,1.965)
\rput[bl](1.11,2.165){\small{Encoding}}
\rput[bl](9.11,2.165){\small{Encryption}}
\psline[linecolor=black, linewidth=0.02, linestyle=dashed, dash=0.17638889cm 0.10583334cm](6.61,3.565)(6.61,0.565)
\psdots[linecolor=black, fillstyle=solid, dotstyle=o, dotsize=0.2, fillcolor=white](0.41,-0.535)
\psdots[linecolor=black, dotsize=0.2](1.41,-0.535)
\psdots[linecolor=black, fillstyle=solid, dotstyle=o, dotsize=0.2, fillcolor=white](2.41,-0.535)
\psdots[linecolor=black, fillstyle=solid, dotstyle=o, dotsize=0.2, fillcolor=white](3.81,-0.535)
\psdots[linecolor=black, dotsize=0.2](4.81,-0.535)
\psdots[linecolor=black, fillstyle=solid, dotstyle=o, dotsize=0.2, fillcolor=white](5.81,-0.535)
\psdots[linecolor=black, fillstyle=solid, dotstyle=o, dotsize=0.2, fillcolor=white](7.97,-0.535)
\psdots[linecolor=black, fillstyle=solid, dotstyle=o, dotsize=0.2, fillcolor=white](0.41,-1.335)
\psdots[linecolor=black, dotsize=0.2](1.41,-1.335)
\psdots[linecolor=black, fillstyle=solid, dotstyle=o, dotsize=0.2, fillcolor=white](2.41,-1.335)
\psdots[linecolor=black, fillstyle=solid, dotstyle=o, dotsize=0.2, fillcolor=white](3.81,-1.335)
\psdots[linecolor=black, dotsize=0.2](4.81,-1.335)
\psdots[linecolor=black, fillstyle=solid, dotstyle=o, dotsize=0.2, fillcolor=white](5.81,-1.335)
\psdots[linecolor=black, fillstyle=solid, dotstyle=o, dotsize=0.2, fillcolor=white](7.97,-1.335)
\psdots[linecolor=black, fillstyle=solid, dotstyle=o, dotsize=0.2, fillcolor=white](0.41,-2.935)
\psdots[linecolor=black, dotsize=0.2](1.41,-2.935)
\psdots[linecolor=black, fillstyle=solid, dotstyle=o, dotsize=0.2, fillcolor=white](2.41,-2.935)
\psdots[linecolor=black, fillstyle=solid, dotstyle=o, dotsize=0.2, fillcolor=white](3.81,-2.935)
\psdots[linecolor=black, dotsize=0.2](4.81,-2.935)
\psdots[linecolor=black, fillstyle=solid, dotstyle=o, dotsize=0.2, fillcolor=white](5.81,-2.935)
\psdots[linecolor=black, fillstyle=solid, dotstyle=o, dotsize=0.2, fillcolor=white](7.97,-2.935)
\psdots[linecolor=black, dotsize=0.2](8.97,-0.535)
\psdots[linecolor=black, dotsize=0.2](8.97,-1.335)
\psdots[linecolor=black, dotsize=0.2](8.97,-2.935)
\psdots[linecolor=black, fillstyle=solid, dotstyle=o, dotsize=0.2, fillcolor=white](9.97,-0.535)
\psdots[linecolor=black, fillstyle=solid, dotstyle=o, dotsize=0.2, fillcolor=white](9.97,-1.335)
\psdots[linecolor=black, fillstyle=solid, dotstyle=o, dotsize=0.2, fillcolor=white](9.97,-2.935)
\rput[bl](8.77,-0.335){$S_n$}
\rput[bl](1.21,-0.335){$S_1$}
\rput[bl](4.64,-0.335){$S_2$}
\psline[linecolor=black, linewidth=0.04, linestyle=dotted, dotsep=0.10583334cm](0.41,-1.735)(0.41,-2.535)
\psline[linecolor=black, linewidth=0.04, linestyle=dotted, dotsep=0.10583334cm](1.41,-1.735)(1.41,-2.535)
\psline[linecolor=black, linewidth=0.04, linestyle=dotted, dotsep=0.10583334cm](2.41,-1.735)(2.41,-2.535)
\psline[linecolor=black, linewidth=0.04, linestyle=dotted, dotsep=0.10583334cm](3.81,-1.735)(3.81,-2.535)
\psline[linecolor=black, linewidth=0.04, linestyle=dotted, dotsep=0.10583334cm](4.81,-1.735)(4.81,-2.535)
\psline[linecolor=black, linewidth=0.04, linestyle=dotted, dotsep=0.10583334cm](5.81,-1.735)(5.81,-2.535)
\psline[linecolor=black, linewidth=0.04, linestyle=dotted, dotsep=0.10583334cm](7.97,-1.735)(7.97,-2.535)
\psline[linecolor=black, linewidth=0.04, linestyle=dotted, dotsep=0.10583334cm](8.97,-1.735)(8.97,-2.535)
\psline[linecolor=black, linewidth=0.04, linestyle=dotted, dotsep=0.10583334cm](9.97,-1.735)(9.97,-2.535)
\psline[linecolor=black, linewidth=0.04, linestyle=dotted, dotsep=0.10583334cm](6.21,-1.735)(7.61,-1.735)(7.61,-1.735)
\psline[linecolor=black, linewidth=0.04, linestyle=dotted, dotsep=0.10583334cm](0.71,-0.535)(1.11,-0.535)
\psline[linecolor=black, linewidth=0.04, linestyle=dotted, dotsep=0.10583334cm](1.71,-0.535)(2.01,-0.535)
\psline[linecolor=black, linewidth=0.04, linestyle=dotted, dotsep=0.10583334cm](4.11,-0.535)(4.51,-0.535)
\psline[linecolor=black, linewidth=0.04, linestyle=dotted, dotsep=0.10583334cm](5.11,-0.535)(5.51,-0.535)
\psline[linecolor=black, linewidth=0.04, linestyle=dotted, dotsep=0.10583334cm](4.11,-1.335)(4.51,-1.335)
\psline[linecolor=black, linewidth=0.04, linestyle=dotted, dotsep=0.10583334cm](5.11,-1.335)(5.51,-1.335)
\psline[linecolor=black, linewidth=0.04, linestyle=dotted, dotsep=0.10583334cm](0.71,-1.335)(1.11,-1.335)
\psline[linecolor=black, linewidth=0.04, linestyle=dotted, dotsep=0.10583334cm](1.71,-1.335)(2.11,-1.335)
\psline[linecolor=black, linewidth=0.04, linestyle=dotted, dotsep=0.10583334cm](0.71,-2.935)(1.11,-2.935)
\psline[linecolor=black, linewidth=0.04, linestyle=dotted, dotsep=0.10583334cm](1.71,-2.935)(2.01,-2.935)
\psline[linecolor=black, linewidth=0.04, linestyle=dotted, dotsep=0.10583334cm](4.11,-2.935)(4.41,-2.935)
\psline[linecolor=black, linewidth=0.04, linestyle=dotted, dotsep=0.10583334cm](5.11,-2.935)(5.41,-2.935)
\psline[linecolor=black, linewidth=0.04, linestyle=dotted, dotsep=0.10583334cm](8.31,-0.535)(8.61,-0.535)
\psline[linecolor=black, linewidth=0.04, linestyle=dotted, dotsep=0.10583334cm](9.31,-0.535)(9.71,-0.535)
\psline[linecolor=black, linewidth=0.04, linestyle=dotted, dotsep=0.10583334cm](8.31,-1.335)(8.61,-1.335)
\psline[linecolor=black, linewidth=0.04, linestyle=dotted, dotsep=0.10583334cm](9.31,-1.335)(9.61,-1.335)
\psline[linecolor=black, linewidth=0.04, linestyle=dotted, dotsep=0.10583334cm](8.31,-2.935)(8.71,-2.935)
\psline[linecolor=black, linewidth=0.04, linestyle=dotted, dotsep=0.10583334cm](9.31,-2.935)(9.61,-2.935)
\psline[linecolor=black, linewidth=0.02](0.21,-3.135)(0.21,-3.235)(1.31,-3.235)(1.41,-3.335)(1.51,-3.235)(2.61,-3.235)(2.61,-3.135)
\rput[bl](1.27,-3.535){$m$}
\psline[linecolor=black, linewidth=0.02](0.01,-3.535)(0.01,-3.735)(5.21,-3.735)(5.41,-3.935)(5.61,-3.735)(10.21,-3.735)(10.21,-3.535)
\rput[bl](5.33,-4.155){$n$}
\psline[linecolor=black, linewidth=0.02](3.27,0.625)(1.61,-0.395)
\psline[linecolor=black, linewidth=0.02](1.6,-0.405)(1.72,-0.405)
\psline[linecolor=black, linewidth=0.02](1.605,-0.41)(1.605,-0.29)
\psline[linecolor=black, linewidth=0.02](4.245,0.595)(4.65,-0.395)
\psline[linecolor=black, linewidth=0.02](4.65,-0.41)(4.695,-0.32)
\psline[linecolor=black, linewidth=0.02](4.665,-0.41)(4.56,-0.365)
\psline[linecolor=black, linewidth=0.02](6.345,0.655)(8.76,-0.44)
\psline[linecolor=black, linewidth=0.02](8.76,-0.44)(8.715,-0.35)
\psline[linecolor=black, linewidth=0.02](8.768,-0.435)(8.69,-0.475)
\end{pspicture}
}
}
\caption{A diagram illustrating the code-based QHE scheme for an $(n+1)$-length 1-fold quantum code while withholding a single subsystem. The $(n+1)$-th subsystem remains in the hands of Client.  The arrows connecting the subsystems indicate where each subsystem (i.e. column) is being mapped.  The filled dots represent code qubits, while the empty dots represent maximally mixed qubits.}
\label{pic}
\end{figure}
The total size of our encrypted input is $mnp$ qubits.  In our preceding notation, the described scheme has parameters $(p, mnp, mnp, \epsilon(m,p),0)$, implementing the set of gates $\mathcal{T}_C$ homomorphically.

\begin{Lemma}
Let $\mathcal{E}$ be the encryption scheme detailed in Figure \ref{qhe}.  Let $\mathcal{T}_C$ denote the group of transversal operators associated to the underlying quantum code $C$.  Then, $\mathcal{E}$ is $\mathcal{T}_C$-homomorphic.
\end{Lemma}

\begin{proof}
Let $U_L$ be the logical operator we wish to apply to some codestate $\ket{\psi}_L$.  By definition, $U_L \in \mathcal{T}_C$ implies $U_L$ can be decomposed as a product operator $U_1 \otimes \ldots \otimes U_{n+r}$ where $U_i$ is an operator that acts only on the $i$-th subsystem of the code.  Then, without knowledge of the secret key $S$, a third party can implement $U_L$ by applying the operator
\[ \bigotimes\limits_{i=1}^{n}\bigotimes\limits_{j=1}^m U_i \]
where each $U_i$ is an operator local to some subsystem in Server's possession (that is to say, on one of the columns in the corresponding array).  Returning the resulting data to a party with the secret key, that party can decrypt to obtain a state of the form $V^RU_L\ket{\psi}_L$, where $V^R$ is supported on the $r$ subsystems that Client has withheld.  Since $r<d$, viewing $V^R$ as an erasure error on $r$ subsystems, there exists some recovery channel $\mathcal{R}$ such that $\mathcal{R}(V^RU_L\ket{\psi_L}) = U_L\ket{\psi}_L$.  Decoding, we obtain $U\ket{\psi}$ as desired. \\
\end{proof}

Note that this scheme is a $\mathcal{T}_C$-homomorphic, non-leveled, and compact QHE scheme, since the recovery and decryption channel do not depend on the complexity of $U$.  

We now aim to compute the security $\epsilon(m,p)$ of the proposed scheme, namely the tradeoff between the size of the input $p$, the size of the encoding $mnp$, and the ITS guarantee.  To avoid confusion, we point out here that the code size $n$ is a constant, as we are not concatenating to achieve security, just amplifying the size of the noise into which we are embedding.  

We want to show that while the scheme is inefficient, its parameters still defeat Nayak's bound. To simplify the security proof, we use the more stringent requirement that the outputs are indistinguishable from uniformly random noise. Here, we will see that the nonlocality of the information stored in QECCs is essential in its allowing us to withhold qubits while still delegating computation to Server. This imposes the requirement of using quantum error-correcting codes, as evidenced by the following observation.

\begin{Lemma}
Suppose we replace the preceding scheme with one that does not withhold any of the physical qubits comprising the (pure state) code.  Then if $m = o(2^p)$, $\epsilon$ must be bounded away from zero.
\end{Lemma}

\begin{proof}
Counting the rank of the encrypted state, note that rank$(\gamma_S) = 2^{np(m-1)}$.  Then,
\begin{align*}
\text{rank}(E_S[\gamma_S]) & \leq m^n 2^{np(m-1)} \\
& \leq 2^{n(p(m-1)+\log(m))}.
\end{align*}
Thus, the fraction of nonzero eigenvalues must be at most $(2^n)^{\log(m) - p}$.  Since $\log(m) = o(p)$, the fraction of nonzero eigenvalues goes to zero, and so $\|E_S[\gamma_S] - I/2^{mnp}\|_1$ must be bounded away from zero as claimed.
\end{proof}

\section{Security tradeoff for the QHE coding scheme}

Our aim is to give (inefficient, but sufficient) security parameters for the coding QHE scheme.  We will then argue that if there were a QECC implementing a sufficiently large transversal gate set (such as the set of all classical reversible gates), then it would violate Nayak's bound with these parameters.  

We will first need a small lemma on the structure of the partial trace operator. The proof can be found in Appendix \ref{Appendix1}.

\begin{Lemma}\label{mixed}
\normalfont
For Hilbert space decomposition $\mathcal{H} = \mathcal{H}_{\bar{\Delta}_1} \otimes \mathcal{H}_{\Delta} \otimes \mathcal{H}_{\bar{\Delta}_2}$,
\[\Tr\left( (\rho^{\bar{\Delta}_1\Delta} \otimes I^{\bar{\Delta}_2})(I^{\bar{\Delta}_1} \otimes \sigma^{\Delta\bar{\Delta}_2}) \right) = \Tr\left( \Tr_{\bar{\Delta}_1}(\rho) \Tr_{\bar{\Delta}_2}(\sigma)\right).\]

\end{Lemma}

With this we are ready to prove the security tradeoff between $\epsilon, p,$ and $m$.  We adopt the same notation used in the proposed scheme for convenience, and note that we are demanding the stronger condition that outputs are indistinguishable from random noise.

\begin{Proposition}
For the scheme described in Figure \ref{qhe}, letting $K = 2^p$ be the dimension of any subsystem and for some $c \in (0,1)$, we have
\[ \| \left(I / K^{mn}\right) - E_S[\gamma_S] \|_1 \leq \epsilon(K,m) \]
for $\epsilon(K,m) = \left( \left(\frac{m-1}{m}\right)^n - 1 + K^{-c}\left(\frac{2K}{m}\right)^n\right)^{1/2}$. \\
\end{Proposition}

\begin{proof}
By Cauchy-Schwartz,
\begin{align*}
\| \left(I/K^{mn}\right) - E_S[\gamma_S] \|_1^2 & \leq K^{mn}\| \left(I/K^{mn}\right) - E_S[\gamma_S]\|_2^2 \\
&\leq K^{mn}\Tr(E_S[\gamma_S]^2) - \left(\frac{2}{K^{mn}}\right)\Tr(E_S[\gamma_S]) + \left(\frac{1}{K^{2(mn)}}\right)\Tr(I) \\
& \leq K^{mn} \Tr(E_S[\gamma_S]^2) - 1.
\end{align*}
where the third line follows by noting that, as a quantum state, $\Tr(E_S[\gamma_S]) = 1$.  We write $|S \cap S'|$ to denote the size of the intersection of $S$ and $S'$ considered as sets.  We can then decompose, for $p_\ell = \Pr_{S,S'} [|S \cap S'| = \ell]$,
\begin{align*}
K^{mn} \Tr(E_S[\gamma_S]^2) &= \left(\frac{K^{mn}}{m^{2n}}\right) \sum\limits_{S,S'} \Tr(\gamma_S\gamma_S') \\
(*) \hspace{1mm} &= K^{mn} \sum\limits_{\ell = 0}^n p_\ell \Tr\left( E[(\gamma_S\gamma_{S'}) \hspace{0.5mm} \big{|} \hspace{1mm} |S\cap S'| = \ell]\right).
\end{align*}
\\
Note that $p_\ell = \frac{ \binom{n}{\ell}(m-1)^{(n-\ell)}}{m^n} \leq \binom{n}{l}/m^\ell$ and that $p_0 = (\frac{m-1}{m})^n$.  Furthermore, up to a permutation on the coordinates, we may write for $\dim(I) = K^{mn-2n}$,
\begin{align*}
K^{mn}E[(\gamma_S\gamma_{S'}) \hspace{0.5mm} \big{|} \hspace{1mm} |S\cap S'| = 0] &= K^{mn} \Tr \left( (\gamma/K^n) \otimes (\gamma/K^n) \otimes \left(I/ K^{(mn-2n)}\right)^2\right) \\
&= 1
\end{align*}
again by noting that $\gamma$ is a quantum state of trace one and by multiplicativity of trace over tensor products.  Next consider the general case $| S \cap S'| = \ell$.  Then up to a permutation on the coordinates and for some $\pi \in S_n$, for $\Delta$ the subsystem of the intersection $S \cap S'$,
\begin{align*}
K^{mn} \Tr(\gamma_S\gamma_{S'}) &= K^{mn} \Tr\left((I/K^{n-\ell} \otimes \gamma)(\pi \gamma \pi^{\dag} \otimes I/K^{n-\ell}) \otimes \left(I/K^{(mn - 2n + \ell)}\right)^2\right)\\
&= K^{\ell} \Tr\left((I \otimes \gamma)(\pi\gamma\pi^\dag \otimes I)\right) \\
&= K^{\ell} \Tr\left(\Tr_{\bar{\Delta}}(\gamma)\Tr_{\bar{\Delta}}(\pi\gamma\pi^\dag)\right)
\end{align*}
where the final line follows from Lemma \ref{mixed}. Then, because we have withheld a subsystem from each subcode of the underlying QECC, in any row $i$ we have that $\Tr_{\bar{\Delta}}(\gamma^i)$ is mixed.  It follows that $\Tr\left(\Tr_{\bar{\Delta}}(\gamma^i)\Tr_{\bar{\Delta}}(\pi\gamma^i\pi^\dag)\right) < 1$. So by separability across each encoded qubit and again by multiplicativity of trace across tensor products,

\[\Tr\left(\bigotimes_{j=1}^p\Tr_{\bar{\Delta}}(\gamma^{i_j})\Tr_{\bar{\Delta}}(\pi\gamma^{i_j}\pi^\dag)\right) = \prod\limits_{j=1}^p \Tr\left(\Tr_{\bar{\Delta}}(\gamma^{i_j})\Tr_{\bar{\Delta}}(\pi\gamma^{i_j}\pi^\dag)\right). \]
It follows that there exists some $c \in (0,1)$ so that 
\[K^{mn} \Tr(\gamma_S\gamma_{S'}) \leq K^{\ell - c}. \]
Putting this all together, we observe that 
\begin{flalign*}
\hspace {1 cm} K^{mn} \sum\limits_{\ell = 1}^n p_\ell \Tr\left( E[(\gamma_S\gamma_{S'}) \hspace{0.5mm} \big{|} \hspace{1mm} |S\cap S'| = \ell]\right) & \leq K^{-c} \sum\limits_{\ell = 1}^n\binom{n}{\ell}\left(\frac{K}{m}\right)^\ell & \\
& \leq K^{-c} \left(\left(1 + \frac{K}{m}\right)^n - 1\right) & \\
& \leq K^{-c} \left(\frac{2K}{m}\right)^n
\end{flalign*}
Including the first term in the sum, we get,

\[ (*) \leq \left(\frac{m-1}{m}\right)^n + K^{-c}\left(\frac{2K}{m}\right)^n \]
and so,

\[\epsilon(K,m) = \left( \left(\frac{m-1}{m}\right)^n - 1 + K^{-c}\left(\frac{2K}{m}\right)^n\right)^{1/2} \]
as desired.

\end{proof}

\section{Limitations on classical transversal computation}

We are left with two competing bounds.  On the one hand, it follows from Nayak's bound (Appendix \ref{AppendixITSQHE}) that, for any $\mathcal{F}$-ITS-QHE encryption scheme with security $\epsilon$ and communication size $s$,
\[ s \geq \log(|\mathcal{F}|) (1 - H(\epsilon)). \]
If we choose parameters that do not leak some constant fraction of information about our input, then as $\epsilon \rightarrow 0$ we see that for $s$ chosen as some fixed function of the input size, it must be that $s = \Omega(\log(|\mathcal{F}|))$.  Using the notation and parameters from the aforementioned coding scheme, this means that $mnp = \Omega(\log(|\mathcal{F}_p|))$ for $\mathcal{F}_p$ the restriction of functions in $\mathcal{F}$ to $p$-bit inputs.  Note that we can assume no ancilla overhead since the constant gets absorbed into this asymptotic bound.  

Now by construction of the scheme, $\mathcal{F}$ is the transversal gate set for the underlying choice of quantum error-correcting code.  Next, we would like to choose $m$ as a function of $K$ so that $\epsilon \rightarrow 0$.  For this, it suffices to choose $m$ as a function of $K$ so that 
\[ \lim_{K \rightarrow \infty} K^{-c}\left(\frac{2K}{m}\right)^n = 0.\]
Equivalently, we require $m = \omega(K^{1-\left(\frac{c}{n}\right)})$.  Then for some $c'< 1$, we can select $m = K^{c'}$ and still have $\epsilon \rightarrow 0$.  Plugging this back into Nayak's bound, we see that asymptotically
\[ K^{c'}\log(K) = \Omega(\log(|\mathcal{F}_p|))\]
for $|\mathcal{F}_p|$ the size of the function class, seen itself as a function returning the number of unique members in the class on $p$-bit inputs.  In particular, $\mathcal{F}_p$  cannot be the set of all Boolean functions, for then $\log(|\mathcal{F}_p|) = K$.  This shows that no code satisfying the hypotheses of our scheme can implement Toffoli transversally.  

We now justify our earlier assumptions on the structure of candidate $r$-fold codes.  Suppose an $r$-fold $[[n,1,d]]$ QECC could implement a logical Toffoli gate transversally. First note that the tensor decomposition between the logical states must align, or else the restriction of logical Toffoli to one element of the product would unitarily map a pure state to a mixed state.  Furthermore, we can think of the QECC criterion in Definition \ref{QECC} as a diagonal and off-diagonal condition: for all $|E| < d$, 
\begin{align*}
\langle 0_L | E | 0_L \rangle &= \langle 1_L | E | 1_L \rangle, \\
\langle 0_L | E | 1_L \rangle &= 0.
\end{align*}

Since the Paulis form an operator basis, we can always assume that $E$ is an element of the Pauli group.  Then, for $r$-fold codes with logical basis states $\ket{i}_L = \bigotimes_{j=1}^r\ket{\psi_{ij}}$, this becomes

\[\prod\limits_{k=1}^r \langle \psi_{ik} | E_k | \psi_{jk} \rangle = c_E \delta_{ij} \]

where $E = E_1 \otimes \ldots \otimes E_r$.  Note then that if $\ket{\psi_{0j}} \not\perp \ket{\psi_{1j}}$, we can trace out the corresponding subsystem and obtain a code with the same correctable error set on the complement of that system.  Furthermore, if $r>d$, then we can again trace out any $r-d$ subcode subsystems to obtain a code with the same correctable error set on the complement.  Both of these observations follow from noticing that these subcodes must themselves satisfy the diagonal condition,

\[ \langle \psi_{0j} | E | \psi_{0j} \rangle = \langle \psi_{1j} | E | \psi_{1j} \rangle. \]

It follows from the security proof that if $r<d$, then the code would satisfy the hypotheses of our scheme and violate the lower bound in Proposition \ref{Impossibility}.  Thus, $r = d$.  Furthermore, logical transversal Toffoli on the entire code must restrict (up to global phase) to a logical transversal Toffoli gate on the subcodes, each of which is $1$-fold by definition.  Thus, each subcode must itself have distance $1$.  To summarize,

\begin{Theorem} \label{Toffoli}
If a QECC is not a maximally redundant code, then it does not admit a classical-reversible universal transversal gate set.  In particular, no such code can implement the Toffoli gate transversally.
\end{Theorem}

Note also that for the scheme in Figure \ref{qhe}, for any $m = \omega(K^{1-(\frac{c}{n})})$, $\epsilon(K)$ is negligible in $p$.  Summarizing the parameters of the coding scheme:

\begin{Proposition}
For any $r$-fold $[[n,1,d]]$ quantum error-correcting code $C$ with $r<d$ and with transversal gate set $\mathcal{T}_C$, the described protocol is a compact quantum $\mathcal{T}_C$-homomorphic encryption scheme with security $\epsilon = negl(p)$ for $p$ the input size and with encoding size $m = 2^{pc'}$ for some $c' < 1$.
\end{Proposition}

While this is highly inefficient, we pause to give some intuition for why it suits our purposes.  On the one hand, we can envision trivial ``hiding'' schemes that have encoding length $2^p$ in each bit.  Nayak's bound allows for higher efficiency, roughly demanding that encodings implementing the set of all classical functions on $p$ bits homomorphically must have length at least $(2^p/p)$ in each bit.  Finally our scheme, with encoding length $2^{pc'}$ for some $c' \in (0,1)$, is just efficient enough to defeat this bound and allow us to argue Theorem \ref{Toffoli}.

Because these maximally redundant codes have a simple design, if we further assume that they are additive, we can use the additional stabilizer structure to argue directly that they cannot implement logical Toffoli transversally.  From this observation, we directly obtain the following.

\begin{Corollary} \label{Additive}
No additive QECC can implement transversal Toffoli.
\end{Corollary}

For proof, see Appendix \ref{Appendix2}.  Note that this also follows from the arguments in Appendix \ref{Appendix3}.  Our central result follows.

\begin{Corollary}
If a QECC is not a non-additive maximally redundant code, then it cannot implement the Toffoli gate transversally.
\end{Corollary}

Finally, note that by concatenating an $[[n,1,d]]$ $d$-fold code with itself, the code remains $d$-fold while the distance must increase to at least $d^2$.  Furthermore, if such a code implements Toffoli strongly transversally, then so does its concatenation with itself.  As a result, we observe the following.

\begin{Corollary}
No QECC can implement strongly transversal Toffoli.
\end{Corollary}

\section{Discussion}

Do there exist non-additive maximally redundant codes that can then implement Toffoli transversally?  One can essentially think of these as QECCs formed by concatenating an outer repetition code with a distance $1$ inner code that is not a stabilizer subspace.  Intuitively, since the inner code is not quantum error-correcting, the code only ``spreads out the information in one basis''.  More precisely, the inner code only satisfies the diagonal QECC criterion.  While this is a less restrictive condition, it still must be ``complementary'' to the outer code, and this allows us to argue impossibility in the additive case. Unfortunately by comparison, the structure of general non-additive codes is less well-understood -- in particular, we know of no examples of such a code.  We expect that {\em no QECC can implement Toffoli transversally}, and view this exception as a consequence of the lack of structure on general non-additive codes.  We hope to resolve this exception in upcoming work.

The QHE scheme we have detailed is non-leveled and compact, but highly inefficient.  An immediate question would be to refine the security proof, which uses too strong a security demand.  It would be most interesting to see if a modified approach can achieve \emph{efficient} ITS-QHE for transversal gate sets of general quantum error-correcting codes, where the size of the encoding is some fixed polynomial of the input length.  There are certain quantitative properties of ``nonlocality'' in QECCs (see e.g. \cite{Arnaud:2013}, \cite{Scott:2004}) that might be helpful in such an endeavor.   Following the outline of \cite{Ouyang:2015}, we could also expect to extend a scheme built on a code with desirable transversal gates to accommodate a constant number of non-transversal gates.  Just as one might tailor a QECC for a specific algorithm that makes heavy use of its transversal gate set, one might also tailor an ITS-QHE scheme to homomorphically implement that algorithm.  Furthermore, it would be of theoretical interest to find a protocol matching the lower bound implicit in Proposition \ref{Impossibility}.

Another interesting open question is to consider leveled ITS-QHE schemes: allow the client some preprocessing to scale with the size of the circuit.  Can this relaxation allow more efficient or universal schemes for polynomial sized circuits, mirroring the computational security case?  A first step might be to try to apply the techniques of instantaneous nonlocal computation \cite{Speelman:2015} that proved invaluable in the computationally secure scheme.  Moreover, through gauge-fixing, we have ways of converting between codes that together form a universal transversal gate set.  Its not clear how to implement such a strategy, since the noisy embedding and non-interactivity present barriers to measuring syndromes, but these elements taken together might be useful in extending the current scheme.

Finally, one could ask if there is a correspondence between transversal gates for quantum codes and nontrivial ITS homomorphically-implementable gate sets, based on the ``richness'' of the function classes they can realize.  In particular, \cite{Eastin:2009} asked: what is the maximum size of finite group that can be implemented logically and transversally?  Indeed, since the Clifford group on $p$-qubits is of size at most $2^{2p^2 + 3p}$ \cite{Ozols:2008}, one could reasonably expect to efficiently implement the Clifford gates homomorphically with information theoretic security, as was done in \cite{Ouyang:2015}.  We hope that our arguments might extend past classical reversible circuit classes to address this question, although it is unclear how to generalize Nayak's bound to apply to these general finite subgroups of the unitary group.

\section{Acknowledgments} 

The authors would like to thank Cupjin Huang and Fang Zhang for useful discussions, in particular concerning the security proof.  They are grateful to Audra McMillan, Anne Broadbent, and Zhengfeng Ji for their comments on an earlier draft of this paper. In particular, we thank Anne for bringing~\cite{Baumeler:2013} to our attention and pointing out that Proposition~\ref{concurrent} follows from the result of that paper and an argument similar to the classical case.  They also thank Fernando Pastawksi for pointing us to Lemma \ref{CliffordLemma}.  This research was supported in part by NSF Awards 1216729, 1318070, and 1526928.
\newpage
\bibliographystyle{abbrv}
\bibliography{bibliography}

\begin{thebibliography}{10}

\bibitem{Ambainis:2008}
A.~Ambainis, D.~Leung, L.~Mancinska, and M.~Ozols.
\newblock Quantum random access codes with shared randomness, October 2008.
\newblock arXiv:0810.2937.

\bibitem{Anderson:2014}
J.~T. Anderson and T.~Jochym-O'Connor.
\newblock Classification of transversal gates in qubit stabilizer codes, 2016.
\newblock Quantum Information and Computation 16, 0771-0802.

\bibitem{Arnaud:2013}
L.~Arnaud and N.~J. Cerf.
\newblock Exploring pure quantum states with maximally mixed reductions,
  January 2013.
\newblock Phys. Rev. A 87, 012319 (2013).

\bibitem{Baumeler:2013}
A.~Baumeler and A.~Broadbent.
\newblock Quantum private information retrieval has linear communication
  complexity, April 2013.
\newblock Journal of Cryptology. Volume 28, Issue 1, pp 161-175 (2015).

\bibitem{Bombin:2013}
H.~Bombin.
\newblock Gauge color codes: Optimal transversal gates and gauge fixing in
  topological stabilizer codes, August 2015.
\newblock New J. Phys. 17 (2015) 083002.

\bibitem{Bravyi:2013}
S.~Bravyi and R.~Koenig.
\newblock Classification of topologically protected gates for local stabilizer
  codes, 2013.
\newblock Phys. Rev. Lett. 110, 170503 (2013).

\bibitem{Bremner:2010}
M.~J. Bremner, R.~Jozsa, and D.~J. Shepherd.
\newblock Classical simulation of commuting quantum computations implies
  collapse of the polynomial hierarchy, August 2010.
\newblock Proceedings of the Royal Society A, Volume 467, Issue 2126.

\bibitem{Broadbent:2015}
A.~Broadbent.
\newblock Delegating private quantum computations, June 2015.
\newblock Canadian Journal of Physics, 2015, 93(9): 941-946.

\bibitem{Broadbent:2008}
A.~Broadbent, J.~Fitzsimons, and E.~Kashefi.
\newblock Universal blind quantum computation, 2009.
\newblock Proceedings of the 50th Annual IEEE Symposium on Foundations of
  Computer Science (FOCS 2009), pp. 517-526.

\bibitem{Broadbent:2014}
A.~Broadbent and S.~Jeffery.
\newblock Quantum homomorphic encryption for circuits of low \uppercase{T}-gate
  complexity, 2015.
\newblock In Proceedings of Advances in Cryptology -- CRYPTO 2015, pp 609-629.

\bibitem{Broadbent:2016}
A.~Broadbent, Z.~Ji, F.~Song, and J.~Watrous.
\newblock Zero-knowledge proof systems for \uppercase{QMA}, April 2016.
\newblock Proceedings of the 2016 IEEE 57th Annual Symposium on Foundations of
  Computer Science (FOCS 2016) pp.31-40.

\bibitem{Cramer:2015}
R.~Cramer, L.~Ducas, C.~Peikert, and O.~Regev.
\newblock Recovering short generators of principal ideals in cyclotomic rings.
\newblock Cryptology ePrint Archive, Report 2015/313, 2015.

\bibitem{Cross:2008}
A.~W. Cross.
\newblock Fault-tolerant quantum computer architectures using hierarchies of
  quantum error-correcting codes, 2008.
\newblock PhD Thesis.

\bibitem{Dulek:2016}
Y.~Dulek, C.~Schaffner, and F.~Speelman.
\newblock Quantum homomorphic encryption for polynomial-sized circuits, August
  2016.
\newblock CRYPTO 2016: Advances in Cryptology - CRYPTO 2016, pp 3-32.

\bibitem{Eastin:2009}
B.~Eastin and E.~Knill.
\newblock Restrictions on transversal encoded quantum gate sets, July 2009.
\newblock Phys. Rev. Lett. 102, 110502.

\bibitem{Fillinger:2012}
M.~Fillinger.
\newblock Lattice based cryptography and fully homomorphic encryption, 2012.
\newblock
  \url{http://homepages.cwi.nl/~schaffne/courses/reports/MaxFillinger_FHE_2012.pdf}.

\bibitem{Fisher:2013}
K.~Fisher, A.~Broadbent, L.~Shalm, Z.~Yan, J.~Lavoie, R.~Prevedel,
  T.~Jennewein, and K.~Resch.
\newblock Quantum computing on encrypted data, January 2014.
\newblock Nature Communications 5, Article number: 3074.

\bibitem{Fowler:2013}
A.~G. Fowler, S.~J. Devitt, and C.~Jones.
\newblock Surface code implementation of block code state distillation, January
  2013.
\newblock Scientific Reports 3, 1939.

\bibitem{Fowler:2012}
A.~G. Fowler, M.~Mariantoni, J.~M. Martinis, and A.~N. Cleland.
\newblock Surface codes: Towards practical large-scale quantum computation,
  August 2012.
\newblock Phys. Rev. A 86, 032324 (2012).

\bibitem{Gentry:2009}
C.~Gentry.
\newblock A fully homomorphic encryption scheme, 2009.
\newblock Ph.D. Thesis, Stanford University.

\bibitem{Gottesman:1997}
D.~Gottesman.
\newblock Stabilizer codes and quantum error correction, 1997.
\newblock Caltech Ph.D. Thesis.

\bibitem{Jones:2013}
C.~Jones.
\newblock Composite toffoli gate with two-round error detection, March 2013.
\newblock Phys. Rev. A 87, 052334.

\bibitem{Jones:2012}
C.~Jones.
\newblock Novel constructions for the fault-tolerant toffoli gate, 2013.
\newblock Phys. Rev. A 87, 022328.

\bibitem{Knill:1996}
E.~Knill and R.~Laflamme.
\newblock A theory of quantum error-correcting codes, April 1996.
\newblock Phys.Rev.Lett.84:2525-2528.

\bibitem{Lai:2017}
C.-Y. Lai and K.-M. Chung.
\newblock On statistically-secure quantum homomorphic encryption.
\newblock In preparation.

\bibitem{Poulsen:2016}
H.~P. Nautrup, N.~Friis, and H.~J. Briegel.
\newblock Topological code switching in two dimensions, September 2016.
\newblock arXiv:1609.08062.

\bibitem{Nayak:1999}
A.~Nayak.
\newblock Optimal lower bounds for quantum automata and random access codes,
  April 1999.
\newblock FOCS 1999.

\bibitem{Nielsen:2011}
M.~A. Nielsen and I.~L. Chuang.
\newblock Quantum computation and quantum information, 2011.
\newblock Cambridge University Press New York, NY.

\bibitem{Ouyang:2015}
Y.~Ouyang, S.-H. Tan, and J.~Fitzsimons.
\newblock Quantum homomorphic encryption from quantum codes, August 2015.
\newblock arXiv:1508.00938.

\bibitem{Ozols:2008}
M.~Ozols.
\newblock Notes on the clifford group, July 2008.
\newblock
  \url{http://home.lu.lv/~sd20008/papers/essays/Clifford%20group%20[paper].pdf}.

\bibitem{Paetznick:2013}
A.~Paetznick and B.~W. Reichardt.
\newblock Universal fault-tolerant quantum computation with only transversal
  gates and error correction, April 2013.
\newblock Phys. Rev. Lett. 111, 090505 (2013).

\bibitem{Pastawski:2014}
F.~Pastawksi and B.~Yoshida.
\newblock Fault-tolerant logical gates in quantum error-correcting codes, 2015.
\newblock Phys. Rev. A 91, 012305.

\bibitem{Peikert:2015}
C.~Peikert.
\newblock A decade of lattice cryptography, March 2016.
\newblock Foundations and Trends in Theoretical Computer Science 10(4):283-424.

\bibitem{Scott:2004}
A.~J. Scott.
\newblock Multipartite entanglement, quantum-error-correcting codes, and
  entangling power of quantum evolutions, May 2004.
\newblock Phys. Rev. A 69, 052330.

\bibitem{Speelman:2015}
F.~Speelman.
\newblock Instantaneous non-local computation of low \uppercase{T}-depth
  quantum circuits, November 2015.
\newblock arXiv:1511.02839.

\bibitem{Tan:2014}
S.-H. Tan, J.~A. Kettlewell, Y.~Ouyang, L.~Chen, and J.~F. Fitzsimons.
\newblock A quantum approach to homomorphic encryption, 2016.
\newblock Sci. Rep. 6, 33467.

\bibitem{Yoder:2016}
T.~J. Yoder, R.~Takagi, and I.~L. Chuang.
\newblock Universal fault-tolerant gates on concatenated stabilizer codes,
  March 2016.
\newblock Phys. Rev. X 6, 031039 (2016).

\bibitem{Yu:2014}
L.~Yu, C.~A. Perez-Delgado, and J.~F. Fitzsimons.
\newblock Limitations on information theoretically secure quantum homomorphic
  encryption, June 2014.
\newblock Phys. Rev. A 90, 050303 (2014).

\bibitem{Zeng:2007}
B.~Zeng, A.~Cross, and I.~L. Chuang.
\newblock Transversality versus universality for additive quantum codes,
  September 2011.
\newblock IEEE Transactions on Information Theory, Volume: 57, Issue: 9, 6272 -
  6284.

\end{thebibliography}
\newpage
\appendix

\section{A no-go result for ITS-QFHE} \label{AppendixITSQHE}

We now set out to show that efficient ITS-QFHE is impossible.  Define for any $x \in \{0,1\}^n$ the state
\[\ket{s_x}^{\mathcal{CM}} = QHE.Enc(x).\]
Note that for any $\ket{s_x}\bra{s_x}^{\mathcal{M}}$, we can always purify to a system of size at most $2|\mathcal{M}|$.  So without loss of generality, we may assume that $|\mathcal{C}| = m$, the size of the message sent from Client to Server.  Next, by information theoretic security, the state of the encryption on subsystem $\mathcal{M}$ must be almost independent of $x$.  Formally, 
\[ \|\Tr_{\mathcal{C}}(\ket{s_x}\bra{s_x}) - \Tr_{\mathcal{C}}(\ket{s_{0^n}}\bra{s_{0^n}}) \|_1 \leq \epsilon \] 
for $\epsilon$ the security of the scheme. Equivalently, there exists a $V_x^\mathcal{C}$ so that, defining $s_x' = (V_x^\mathcal{C} \otimes I^\mathcal{M}) \ket{s_{0^n}}\bra{s_{0^n}} (V_x^\mathcal{C} \otimes I^\mathcal{M})^\dag$ and $s_x = \ket{s_x}\bra{s_x}$,
\[\|s_x' - s_x\|_1 \leq \epsilon.\]
Furthermore, for any $f \in \mathcal{F}$, we have by the homomorphic property that, abbreviating $QHE.Eval_f$ as $f_{ev}(\cdot)$ and $QHE.Dec(\cdot)$ as $D(\cdot)$,
\[\left[I^{\mathcal{C}} \otimes f_{ev}^\mathcal{M}(s_x^{\mathcal{CM}}) \right]^\mathcal{CM'} \eqqcolon \eta_{f,x},\]
\[D^\mathcal{CM'}(\eta_{f,x}) = f(x).\]
But now, defining $\eta_{f,x}'$ by replacing $s_x$ with $s_x'$ in the definition of $\eta_{f,x}$, by contractivity of trace distance we also have
\[\Pr[D^\mathcal{CM'}(\eta_{f,x}') \neq f(x)] \leq \epsilon. \]
To elucidate the underlying QRAC, define the mapping
\[f \mapsto \eta_{f,0^n},\]
and note that $(V_x^\mathcal{C} \otimes I^\mathcal{M'})\eta_{f,0^n}(V_x^\mathcal{C} \otimes I^\mathcal{M'})^\dag = \eta'_{f,x}$.  So let $D^\mathcal{CM'}\left[(V_x^\mathcal{C} \otimes I^\mathcal{M'})(\cdot) \right]$ denote the query for index $x$ of $f$, thinking of $f$ as a $2^n$ length bit string, with the $x$th bit defined as $f(x)$.  Then we have a $(2^n, m + m',1-\epsilon)$-QRAC for the set of all Boolean functions, where $m+m'$ is the communication cost of the protocol.  We now recall a well-known bound on the efficiency of QRACs \cite{Nayak:1999}.

\begin{Theorem}[Nayak's Bound]
If there exists an $(n,m,p)$-QRAC, then for $H(\cdot)$ the binary entropy function,
\[ m \geq n(1 - H(p)). \]
\end{Theorem}

So it must be that the total communication cost of the protocol $(m+m') \geq 2^n(1 - H(\epsilon))$.  For security, allowing $\epsilon \rightarrow 0$ and noting that $H(\epsilon) \rightarrow 0$, we see that the communication cost $(m+m') = \Omega(2^n)$.  Thus, either the size of the encoding or the evaluated ciphertext must be exponentially long in the input, precluding efficiency.  In short,

\begin{Proposition}\label{Impossibility}
The communication cost of ITS-QFHE must be exponential in the size of the input.
\end{Proposition}

\section{Proof of Lemma \ref{mixed}}\label{Appendix1}

\begin{proof}
Expanding in terms of outer products,
\begin{align*}
\Tr\left((\rho \otimes I)(I \otimes \sigma)\right) &= \Tr\Bigg{(} \Big{(}\sum\limits_{i,i'} \sum\limits_{j,j'} \sum\limits_k a_{i,i',j,j'} \ket{i}\bra{i}^{\bar{\Delta}_1} \otimes \ket{j}\bra{j'}^\Delta \otimes \ket{k}\bra{k}^{\bar{\Delta}_2}\Big{)} \cdot \\ 
& \hspace{12 mm} \Big{(}\sum\limits_{\ell} \sum\limits_{m,m'} \sum\limits_{n,n'} b_{m,m',n,n'} \ket{\ell}\bra{\ell}^{\bar{\Delta}_1} \otimes \ket{m}\bra{m'}^\Delta \otimes \ket{n}\bra{n'}^{\bar{\Delta}_2}\Big{)}\Bigg{)} \\
& = \Tr\Bigg{(} \sum\limits_{i,i'} \sum\limits_{n,n'} \sum\limits_{j,m'} \left( \sum\limits_{j'} a_{i,i',j,j'} b_{j',m',n,n'} \right) \ket{i}\bra{i} \otimes \ket{j}\bra{m'} \otimes \ket{n}\bra{n'}\Bigg{)} \\
&= \sum\limits_i \sum\limits_n \sum\limits_{j,j'}\left(a_{i,i,j,j'}b_{j',j,n,n}\right).
\end{align*}
On the other hand, we have
\begin{align*}
\Tr\left( \Tr_{\bar{\Delta}_1}(\rho) \Tr_{\bar{\Delta}_2}(\sigma)\right) & = \Tr\Bigg{(} \Big{(}\sum\limits_i \sum\limits_{j,j'} a_{i,i,j,j'} \ket{j}\bra{j'}\Big{)} \Big{(}\sum\limits_n \sum\limits_{m,m'} b_{m,m',n,n} \ket{m}\bra{m'}\Big{)}\Bigg{)}\\
&= \Tr\Bigg{(}\sum\limits_i \sum\limits_n \sum\limits_{j,j'}\Big{(} \sum\limits_{j'} a_{i,i,j,j'}b_{j',m',n,n}\Big{)}\ket{j}\bra{m'}\Bigg{)}\\
&= \sum\limits_i \sum\limits_n \sum\limits_{j,j'}\left(a_{i,i,j,j'}b_{j',j,n,n}\right)
\end{align*}
as claimed.
\end{proof}

\section{Proof of Corollary \ref{Additive}}\label{Appendix2}

\begin{proof}
By Theorem \ref{Toffoli}, it suffices to consider maximally redundant codes.  So suppose, for the sake of contradiction, that an $[[n,1,d]]$ additive $d$-fold  code could implement Toffoli transversally.  Let $[ \cdot, \cdot ]$ denote the group commutator.  We denote by $\bar{\cdot}$ states and operations acting on the subcodes, and $\tilde{\cdot}$ those on the full code.  We will assume that each subcode is the same, e.g. $\ket{\tilde{i}} = \ket{\bar{i}}^{\otimes d}$, so that we can speak directly about the inner and outer codes. The general argument follows similarly.

Since the code is additive, the code distance is the minimal weight logical Pauli operator acting on the code. For any $\bar{Z}_L$, by multiplicativity of the inner product over tensor products,

\begin{align*}
\frac{1}{2} \left\langle\ket{\tilde{0}} + \ket{\tilde{1}}\left|\bar{Z}_L\right|\ket{\tilde{0}} - \ket{\tilde{1}}\right\rangle &= \frac{1}{2} \left(  \langle \tilde{0}| \bar{Z}_L| \tilde{0} \rangle - \langle \tilde{0}| \bar{Z}_L| \tilde{1} \rangle + \langle \tilde{0}| \bar{Z}_L| \tilde{0} \rangle -\langle \tilde{1}| \bar{Z}_L| \tilde{1} \rangle \right) \\
&= \frac{1}{2}\left(\langle \bar{0}|\bar{0} \rangle^{\frac{n}{d}} + \langle \bar{1}|\bar{1} \rangle^{\frac{n}{d}}\right) \neq 0.
\end{align*}
Since the outer code has distance $d$, it follows from the QECC criterion that $\bar{Z}_L$ must have weight at least $d$.  Then $\bar{X}_L$ must have weight $1$, since the underlying inner code has distance $1$ by assumption.  Because the outer classical repetition code factors as a tensor product, transversal $\widetilde{\text{Toff}}_L$ on the outer code must restrict (up to a global phase) to transversal $\overline{\text{Toff}}_L$ on the inner code. Since we're now working with multiqubit gates, let $G_L(i)$ denote the logical gate for $G$ acting on the $i$th code block.  We can compute directly,

\[ [\overline{\text{Toff}}_L(1,2,3), \bar{X}_L(1)] = \overline{CX}_L(2,3).\]

Furthermore, because $\overline{\text{Toff}}_L$ and $\bar{X}_L$ are transversal, it follows that $\overline{CX}_L$ has a representative that is also transversal and is supported on the subsystems that support $\bar{X}_L$.  By a similar argument

\[ [\overline{CX}_L(1,2), \bar{Z}_L(1)] = \bar{Z}_L(2) \]

so that $\bar{Z}_L$ must also be contained in the subsystems supporting $\overline{CX}_L$, and in turn $\bar{X}_L$.  As we have already observed, the minimal weight of any representative of $\bar{Z}_L$ must be at least $d$, a contradiction as $\bar{X}_L$ has a representative of weight $1$.

\end{proof}

\section{An alternate proof for stabilizer codes} \label{Appendix3}

Here we offer an alternate proof limiting universal transversal reversible computation for the subclass of stabilizer codes.  The arguments here are based off of results from \cite{Bravyi:2013} and \cite{Pastawski:2014}, which we reproduce for completeness.

\begin{Definition} \normalfont
The \emph{Clifford hierarchy} $\mathcal{C}$ is a sequence of gate sets $\{\mathcal{C}_k\}_{k \geq 1}$ defined recursively by $\mathcal{C}_k = \{U: U\mathcal{C}_1U^\dag \subseteq \mathcal{C}_{k-1}\}$, where we define $\mathcal{C}_1$ to be the Pauli group.  
\end{Definition}

Note that $\mathcal{C}_2$ is the Clifford group, and $\mathcal{C}_k$ fails to be a group for $k>2$.  Further note that reversible circuits saturate the Clifford hierarchy (and in fact can lie outside it entirely) by the gate $C^kX$, the $k$-controlled bit-flip gate, which lies in $\mathcal{C}_{k+1}$.  Toffoli is simply $C^2X$, and so lies in the third level of the Clifford hierarchy.  We next recall the stabilizer cleaning lemma, which can be found in \cite{Bravyi:2013}.

\begin{Lemma} \normalfont
Let $S$ be a stabilizer code, and let $R$ be any subset of physical qubits of the code such that any logical operator supported on $R$ acts trivially on $S$.  Then, for any logical operator $U_L$, there exists a representative of $U_L$ supported on $R^c$.
\end{Lemma}

We call such subsets $R$ cleanable.  Equipped with the cleaning lemma, we can now summarize the following lemma from \cite{Pastawski:2014}.

\begin{Lemma} \label{CliffordLemma} \normalfont
Let $S$ be a stabilizer code and let $\{R_0, \ldots, R_k\}$ be a set of cleanable subsets of the physical qubits comprising $S$.  Let $U$ be a logical operator supported on $\cup_{i=0}^k R_i$ such that $U$ is transversal with respect to the $R_i$.  Then, $U_L \in \mathcal{C}_k$.
\end{Lemma}

\begin{proof}
We proceed by induction on $k$.  In the base case, we have a logical operator $U$ supported on cleanable subsets $R_0 \cup R_1$.  Let $P$ be any logical Pauli operator cleaned off of $R_1$, and let $[ \cdot, \cdot ]$ denote the group commutator.  Since in a stabilizer code the logical Pauli operators are transversal, we have $Supp([U,P]) \subseteq R_0$, which by cleanability implies that $[U_L,P_L] = cI_L$.  Since this is true for any $P_L$, it must be that $U_L \in \mathcal{C}_1$.

Similarly, suppose $U$ is supported on $\cup_{i=0}^k R_i$.  Then, cleaning any logical Pauli $P$ off of $R_k$, we see that $Supp([U,P]) \subseteq \cup_{i=0}^{k-1} R_i$.  By our inductive hypothesis, $[U_L, P_L] \subseteq \mathcal{C}_{k-1}$, which implies $U_LP_LU_L^\dag \in \mathcal{C}_{k-1}$ for any logical Pauli $P_L$.  Thus $U_L \in \mathcal{C}_k$, completing the proof.
\end{proof}

This argument generalizes to subsystem codes, and we refer the reader \cite{Pastawski:2014} for a more complete description.  As a consequence we obtain the following.

\begin{Corollary}
No erasure-correcting stabilizer code can implement a classical reversible universal transversal gate set.
\end{Corollary}

\begin{proof}
Partition the code block into single subsystem subsets $\{R_1 , \ldots, R_n\}$ where $n$ is the length of the code.  Then, since the code is erasure-correcting, any logical operator supported on a single subsystem must act trivially on the codespace, and so these subsets are cleanable.  By the lemma, any transversal logical gate must lie in $\mathcal{C}_n$.  Since reversible circuits saturate $\mathcal{C}$, they cannot be logically transversally implementable.
\end{proof}

\end{document}